\documentclass{article}

\usepackage{amsmath,amsthm,amssymb,natbib,xcolor,verbatim}

	\theoremstyle{plain}
	\newtheorem{theorem}{Theorem}

	\newtheorem*{axiom*}{Axiom}

	\newtheorem{lemma}{Lemma}

	\theoremstyle{definition}
	
	\newtheorem{example}{Example}

	\theoremstyle{remark}
	\newtheorem*{remark}{Remark}
	
	\newtheorem*{claim*}{Claim}
	
\title{A Note on Invariant Extensions of Preorders}

\author{Peter Caradonna and Christopher P. Chambers}

\begin{document}
\maketitle
\begin{abstract}
    We consider the problem of extending an acyclic binary relation that is invariant under a given family of transformations into an invariant preference.  We show that when a family of transformations is \emph{commutative}, every acyclic invariant binary relation extends.  We find that, in general, the set of extensions agree on the ranking of many pairs that (i) are unranked by the original relation, and (ii) cannot be ranked by invariance or transitivity considerations alone.  We interpret these additional implications as the out-of-sample predictions generated by invariance, and study their structure. 
\end{abstract}

\medskip

\section{Introduction}

The aim of this note is to study the out-of-sample predictions generated by various models of preference under limited data.  We  specifically consider families of preferences that are \emph{invariant} or preserved under some collection of transformations of the underlying consumption space.  We are interested in the counterfactual, or out-of-sample predictions that such models of preference generate from limited or incomplete data. \medskip

Formally, given a consumption space $X$, we consider a collection $\mathcal{M}$ of transformations $X \to X$.  We are interested in the problem of when a given binary relation $\succeq$ on $X$ may be extended into a weak order that is invariant under the transformations in $\mathcal{M}$, in the sense that $x \succeq y \iff \omega(x) \succeq \omega(y)$ for all $\omega \in \mathcal{M}$.\footnote{Note that when $\mathcal{M}$ contains a single element, the identity function $\textrm{id}: X \to X$, our question reduces to the classical problem of extending a binary relation to a weak order.} \medskip

This exercise is intimately related to the problem of characterizing the empirical content of such models.  Classically, the falsification of a theory requires some subset of observations to generate implications which are observed to be false elsewhere in the data.  In the case of rational choice, it is well-understood that the only testable implications are given by the transitive closure of the revealed preference relation.  In this context, we find that, when one requires rationalizability by an \emph{invariant} preference relation, even simple, finite data sets generate significantly richer out-of-sample predictions.  \medskip

\subsection{Related Literature}

A classic work in this direction is \citet{dubra2004}.  In this framework, individual choice is over lotteries, and ``rationality'' is interpreted as the satisfaction of the independence axiom (in addition to transitivity).  Among other things, this paper establishes that satisfaction of the independence axiom on a small set of lotteries is always consistent with satisfaction of the independence axiom on the set of all lotteries.\medskip

Another classic reference in this vein is \citet{duggan}, who retains an abstract framework but imposes additional restrictions on the interpretation of ``rationality.''  He establishes that a relatively broad class of notions of rationality all lead to the same idea.  So long as rationality is satisfied on a subset of alternatives, then we cannot falsify the hypothesis of its satisfaction on the set of all alternatives.\medskip

\citet{demuynck2009} investigates a very closely related question.  His idea is to study closure operators as applied to preference relations.\footnote{See \citet{ward} for a general theory of closures.}  Roughly, associated with certain algebraic theories of preference, there is a ``smallest'' such preference satisfying that property and containing a given preorder.  He provides a general extension result for algebraic structures satisfying certain properties.  Our work also features a few closure operators, though we do not use these terms.\footnote{In particular, the transitive closure, the smallest $m$-coherent relation containing a given relation, and the lattice-theoretic join of these two closures, the smallest transitive and $m$-coherent relation containing a given relation.} \medskip

His work seems to be the first to establish that a homothetic and monotonic preorder on consumption space possessess a homothetic and monotonic weak order extension.  A direct corollary of our main result is that the same conclusion holds when monotonicity is not required.  \medskip

\citet{demuynck2009} focuses on establishing highly general results:  any property of binary relations which generates a closure can be meaningfully considered, and seeks to find minimal conditions on this closure guaranteeing extension results.  For example, he establishes results on convexity of preference, which our work is not able to discuss.  On the other hand, the general nature of the result means that for each particular algebraic structure, there is work involved in establishing his conditions are satisfied.  Our work focuses on a smaller class of algebraic properties but is able to derive results that are immediately applicable.\medskip

Moving forward, we focus on a relatively new structure which has recently been fruitfully applied to revealed preference theory:  this framework is that of \citet{freer2021,freer2022}, which builds off of \citet{demuynck2009}.  The economic structure in these works is supposed to be algebraic in nature.  The authors hypothesize that a preference satisfies a kind of generalized notion of quasilinearity.  It pays to be a bit formal here.  They imagine that there is a transformation function $m$ which maps each alternative into another alternative, and want to test rationality with respect to preferences that are ``coherent'' with respect to that transformation function.  Their conditions are of the nature of revealed preference conditions, postulating the absence of certain types of cycles.  But the technology of the general notion of $m$-coherency appears first here. \medskip

We have mentioned already the contribution of \citet{dushnik1941}.  Several authors in economics have taken interest in the Pareto representation of partial or preorders.  Abstract approaches include \citet{donaldson1998,bossert1999,weymark2000generalization,alcantud2009}. \medskip

Our work presupposes no  notion of topology, but many works in economics involving extensions take topological considerations very seriously.  \citet{aumann1962utility,aumann1964utility,peleg1970,levin} are classical references, but the theory has developed much since then.  In particular, \citet{ok2002} can now be considered a canonical reference.\footnote{See also the recent work involving further generalizations to non-transitive preferences; \emph{e.g.} \citet{evren2011multi,nishimura2016} and others.} \medskip

In concrete economic environments, similar representations can be found in, for example, the theory of expected utility preferences \citet{dubra2004,gorno2017strict}, Krepsian style  preferences over menus \citet{nehring1999multi}, or rankings of accomplishments \citet{chambersmiller2018}.

\section{Commutative Families}

Let $X$ be a set and let $\mathcal{M}$ be a non-empty collection of transformations $X \to X$.  We say that a binary relation $\succeq$ on $X$ is $\mathcal{M}$-\emph{coherent} if, for all $x,y\in X$ and $\omega \in \mathcal{M}$, 
\[
x \succeq y \quad \implies \quad \omega(x) \succeq \omega(y)
\]
and 
\[
x\succ y \quad \implies \quad \omega(x) \succ \omega(y).
\]
The relation $\succeq$ is \emph{strongly coherent} if, in addition, the converse implication obtains. We say that a relation is \emph{acyclic} if there is no $k\geq 2$ and distinct $x_1,\ldots,x_k\in X$ for which $x_1 \succeq x_2 \succeq \ldots x_k \succ x_1$. \medskip

For a binary relation $\succeq$, an \emph{extension} is a relation $\succeq'$ such that $\succeq \, \subseteq\,  \succeq’$ and $\succ\, \subseteq\,  \succ’$.  A binary relation is a \emph{preorder} if it is reflexive and transitive, and a \emph{partial order} if it is an antisymmetric preorder.  A \emph{weak order} is a complete preorder, and a \emph{linear order} is a complete partial order.\footnote{Reflexive: $x\succeq x$. Complete:  $x\succeq y$ or $y \succeq x$.  Transitive: $x \succeq y$ and $y \succeq z$ implies $x\succeq z$.  Antisymmetric: $x \succeq y$ and $y \succeq x$ implies $x = y$.}  We use the notation $x\parallel y$ when neither $x\succeq y$ nor $y \succeq x$. \medskip

We will call $\mathcal{M}$ a \emph{commutative family} if the following three hypotheses are satisfied:
\begin{enumerate}
\item[(i)] {\bf Commutativity}: For all $\omega,\omega'\in \mathcal{M}$, $\omega \circ \omega'=\omega'\circ \omega$.
\item[(ii)] {\bf Identity}: The map $\mbox{id}:X\rightarrow X$ defined by $\mbox{id}(x) = x$ is a member of $\mathcal{M}$
\item[(iii)] {\bf Closure}: For all $\omega,\omega'\in \mathcal{M}$, their composition $\omega\circ \omega'\in \mathcal{M}$.
\end{enumerate}
Given our consideration of $\mathcal{M}$-coherent binary relations, the only restrictive assumption is commutativity.  If $\mathcal{M}$ does not satisfy (ii) and (iii), but a binary relation is $\mathcal{M}$-coherent, it will remain so when $\mathcal{M}$ is augmented to include the identity and all finite compositions.  Thus it is without loss of generality to suppose that $\mathcal{M}$ is a \emph{monoid} of transformations of $X$; in this section we will suppose, in addition, that this monoid is \emph{commutative}. \medskip

Our first result says that when $\mathcal{M}$ is a commutative family, \emph{every} acyclic, $\mathcal{M}$-coherent relation admits an $\mathcal{M}$-coherent weak order extension.

\begin{theorem}\label{thm:preorderextension}Let $\mathcal{M}$ be a commutative family, and let $\succeq$ be an acyclic, $\mathcal{M}$-coherent binary relation.  Then $\succeq$ has a strongly $\mathcal{M}$-coherent weak order extension.  \end{theorem}

Let us illustrate the result with an example, which does not seem to exist in the literature.

\begin{example}Let $X = \mathbb{R}_+^L$ and for each $\alpha > 0$, let $\omega_{\alpha}(x) = \alpha x$.  Observe that the family $\mathcal{M}=\{\omega_{\alpha}\}_{\alpha > 0}$ is a commutative family.  A preorder $\succeq$ is $\mathcal{M}$-coherent if and only if, for every $x,y\in \mathbb{R}_+^L$ and $\alpha > 0$, we have $x \succeq y $ $\iff$ $\alpha x \succeq \alpha y$.  Such a preorder is called \emph{homothetic}.  Theorem~\ref{thm:preorderextension} demonstrates that an arbitrary homothetic preorder can be extended to a homothetic weak order; obviously this result would continue to hold whenever $X$ is a cone in any real-vector space.  \citet{demuynck2009} establishes that every monotone and homothetic preorder has a monotone and homothetic weak order extension; this is implied by our Theorem~\ref{thm:preorderextension}, by simply letting $\succeq$ be a homothetic relation that contains the usual component-wise ordering of $\mathbb{R}^L_+$.  \end{example}

\begin{example}
Consider an environment where $X$ consists of a family of \emph{dated rewards} $X = Y \times \mathbb{R}_+$ where $Y$ is a set of rewards; relevant references here include \citet{fishburn1982} and \citet{halevy2015}.  A pair $(y,t)$ represents the consumption of $y$ at date $t$.  In this setting, stationarity means that $(y,t) \succeq (y',t')$ implies $(y,t+t'')\succeq (y',t'+t'')$, with a corresponding statement for strict preference.\medskip

Thus, stationarity for dated rewards is $M$-coherency with respect to $\mathcal{M}=\{\omega_t\}_{t\geq 0}$, where $\omega_t(y,t') = (y,t'+t)$.  $\mathcal{M}$ forms a commutative family and thus by Theorem~\ref{thm:preorderextension}, every stationary preorder extends to a stationary weak order. 
\end{example}

\begin{example}We show how to demonstrate a non-topological analogue of \citet{dubra2004}, using Theorem~\ref{thm:preorderextension}.   
Let $(Y,\Sigma)$ be some measurable space and let $\Delta(Y)$ be the set of countably additive probability measures on $(Y,\Sigma)$.  Say that $\succeq$ on $\Delta(Y)$ satisfies \emph{rational independence} if for all $p,q,r\in\Delta(Y)$ and $\alpha \in \mathbb{Q}\cap (0,1]$, $p \succeq q$ if and only if $\alpha p + (1-\alpha)r \succeq \alpha q + (1-\alpha)r$.  \medskip

Extend $\succeq$ to the set of all signed measures of bounded variation as follows.  Say $\nu\succeq^* \nu'$ if and only if there is $\alpha\in\mathbb{Q}$ for which $\alpha > 0$, and $p,q\in\Delta(Y)$ satisfying $p \succeq q$ for which $(\nu-\nu')=\alpha(p-q)$.  \medskip

Observe that if $\nu-\nu'=\alpha(p-q)=\beta(r-s)$ and $p\succeq q$, then it cannot be that $s \succ r$ (by by rational independence and transitivity of $\succeq$).  In other words, $\succeq^*$ is (informally speaking) an extension of $\succeq$.\footnote{This is only informal as $\succeq^*$ and $\succeq$ are defined on different sets.}  Likewise, if $\nu_1\succeq^* \nu_2 \succeq^* \nu_3$, then it is similarly straightforward to establish that $\nu_1 \succeq^* \nu_3$.  \medskip

Now, observe that for $\succeq^*$, we have $\nu \succeq^* \nu'$ if and only if $\nu + \rho \succeq^* \nu'+\rho$, for any signed measure $\rho$.  Let $\mathcal{M}$ consist of all maps $\omega_{\rho}(\nu)=\nu+\rho$, which clearly commute across $\rho$. Thus by Theorem~\ref{thm:preorderextension}, we can extend to $\succeq'$ on all signed measures, preserving $\mathcal{M}$-coherency, as desired.  \medskip

Finally, observe that the restriction of $\succeq'$ to $\Delta(Y)$ satisfies rational independence:  if $p \succeq' q$ and $\alpha \in \mathbb{Q}\cap (0,1]$, then we must have $\alpha p \succeq' \alpha q$ by transitivity and the fact that $\succeq'$ commutes with respect to each $m_{\rho}$.\footnote{This follows from a straightforward induction argument.  For example, let us show that if $p \succeq' q$, then $(1/2)p \succeq' (1/2)q$.  If not, by completeness, $(1/2)q \succ' (1/2)p$.  Then $q = (1/2)q + (1/2)q \succ' (1/2)q + (1/2)p \succ' (1/2)q + (1/2)q$, where each $\succ'$ follows from an application of $m_{(1/2)q}$-coherency.  Transitivity then implies $q \succ' p$, a contradiction.}   Consequently,  then $\alpha p+(1-\alpha)r \succeq' \alpha q + (1-\alpha)r$, with a similar statement holding for strict preference.  \end{example}

\subsection{Counterfactual Predictions}

Even in the case of $\mathcal{M}$-coherency for commutative families, there will generally exist out-of-sample predictions that are not accounted for by simply invariance or transitivity in isolation.

\begin{example}

Let $X = \{a,b\}\times \mathbb{Z}$, with  $m(i,z) = (i,z+1)$. Define $\mathcal{M} = \{\textrm{id}, m, m\circ m, \cdots \}$, and let $\succeq$ be the preorder whose only nontrivial rankings are $(a,z)\succ (b,z+1)$ and $(a,z) \succ (b,z-1)$.  Observe that $\succeq$ is indeed $\mathcal{M}$-coherent for the commutative family $\mathcal{M}$.  Nevertheless, every $\mathcal{M}$-coherent weak order extension $\succeq^*$ must have $(a,0) \succ^* (b,0)$.   To see why, suppose by means of contradiction that there is such an extension whereby $(b,0) \succeq^* (a,0)$.  Then $(b,0) \succeq^* (a,0) \succ^* (b,1) \succeq^* (a,1) \succ^* (b,0)$.  The third ranking is by $\mathcal{M}$-coherency and the second and fourth by the extension property.  This constitutes a violation of transitivity.  Conversely, it is easy to construct $\mathcal{M}$-coherent extensions for which $(a,0) \succ^* (b,0)$.\end{example}

\section{Non-commutative Families}

When commutativity of $M$ is discarded, the conclusion of Theorem~\ref{thm:preorderextension} fails to hold.  As the next example highlights, this is due, roughly, to the possibility of having multiple, mutually inconsistent forcing collections for particular pairs of $\succeq$-incomparable alternatives.

\begin{example} Let $Z = \{a,b,a',b',x,y\}$ denote a space of prizes, and let $X = Z^{\mathbb{N}}$ denote the space of all infinite-horizon consumption streams $\sigma = (\sigma_1,\sigma_2, \ldots)$ taking values in $Z$. For each $z \in Z$, define $\omega_z(\sigma)$ to be the stream obtained by appending the prize $z$ to period one, and shifting all terms in $\sigma$ out by one period:

$$\omega_z(\sigma) = (z, \sigma_1, \sigma_2,\ldots).$$

Let $\mathcal{M}$ denote the free monoid generated by the transformations $\{\omega_z\}_{z \in Z}$.\footnote{That is, an element of $\mathcal{M}$ consists of a finite composition of the transformations $\{\omega_z\}_{z \in Z}$. We may regard such a transformation as a finite string on the alphabet $Z$.} Here, $\mathcal{M}-coherency$ coincides with the stationarity axiom of \citet{koopmans1960}.  Let $\succeq$ be the preorder consisting of the relations:
\begin{equation}\label{cex}
\begin{aligned}
\omega_a(\sigma_x) & \succ \omega_b(\sigma_y)\\
\omega_b(\sigma_x) & \succ \omega_a(\sigma_y)\\
\omega_{a'}(\sigma_x) & \succ \omega_{b'}(\sigma_y)\\
\omega_{b'}(\sigma_x) & \succ \omega_{a'}(\sigma_y)
\end{aligned}
\end{equation}
where $\sigma_x = (x,\ldots)$ and $\sigma_y = (y,\ldots)$ denote the constant $x$ and $y$ streams, as well as all forward and backward `translates' of these relations under elements of $\mathcal{M}$. This relation is trivially transitive and is $\mathcal{M}$-coherent by construction.  However, no $\mathcal{M}$-coherent extension can exist: any such extension $\succeq^*$ must specify the relation between $\sigma_x$ and $\sigma_y$; however, the first two relations in \eqref{cex} prohibit $\sigma_y \succeq^* \sigma_x$, while the latter two prohibit $\sigma_x \succeq^* \sigma_y$.  Thus while $\succeq$ is consistent with Koopmans' stationarity axiom, there is no extension of $\succeq$ to a preference relation that preserves stationarity.
\end{example}

The preceding example demonstrates that in general, a stationary preorder does not have a stationary weak order extension.

\section{Conclusion}

Future research will investigate the further structure for non-commutative families, as well as topological considerations.  Related are the papers by \citet{ok2014topological,ok2021fully}, which study related questions for groups rather than monoids.

\begin{appendix}
\section{Proof of Theorem~\ref{thm:preorderextension}}

\begin{lemma}\label{lem:transitive}Suppose that an acyclic relation satisfies $\mathcal{M}$-coherency.  Then so does its transitive closure.\end{lemma}

\begin{proof}Let $\succeq$ be an acyclic relation satisfying $\mathcal{M}$-coherency.  Define $\succeq^T$ to be its transitive closure.  Let $x,y\in X$ for which $x\succeq^T y$.  Then there is $k$ and $x = x_1 \succeq \ldots \succeq x_k =y$.  Therefore, $\omega(x)=\omega(x_1) \succeq \ldots \succeq \omega(x_k) = \omega(y)$, so that $\omega(x) \succeq^T \omega(y)$ for any $\omega \in \mathcal{M}$.  \medskip

Now suppose that additionally it is not the case that $y\succeq^T x$.  We want to show that it is not the case that $\omega(y) \succeq^T \omega(x)$, for any $\omega \in \mathcal{M}$. \medskip

Since $y\succeq^T x$ is false and $x \succeq^T y$ is true, it follows that there is a chain $x= x_1 \succeq \ldots x_k = y$, where $x_i \succ x_{i+1}$ for some $i,i+1$ (otherwise if for all $i$, $x_i \sim x_{i+1}$, we could follow the chain back and $y \succeq^T x$).  Consequently $\omega(x_1) \succeq \ldots \succeq \omega(x_k)$, where $\omega(x_i) \succ \omega(x_{i+1})$ for some $i,i+1$.  Because $\succeq$ is acyclic, it follows that there is no chain from $\omega(x_k)$ to $\omega(x_1)$, so that in particular $\omega(x_k) \succeq^T \omega(x_1)$ is false. Since $\omega \in \mathcal{M}$ was arbitrary the result follows.\end{proof}

\begin{remark}The transitive closure of $\succeq$ is an extension of $\succeq$ if and only if $\succeq$ is acyclic.  This is a simple consequence of arguments found in \citet{richter1966revealed,richter71,hansson1968choice}, and which are based on \citet{szpilrajn1930}.  See also \citet{suzumura1976}, and \citet{chambers2016revealed}.\end{remark}

\begin{lemma}\label{lem:strongcoherent2}Suppose $\mathcal{M}$ is a commutative family.  Then every $\mathcal{M}$-coherent preorder has a strongly $\mathcal{M}$-coherent preorder extension.  
\end{lemma}

\begin{proof} 
Let $\succeq$ be $\mathcal{M}$-coherent.  Define $\succeq’$ by $x \succeq’ y$ if there exists $\omega_1,\ldots,\omega_k\in \mathcal{M}$ for which $(\omega_1\circ\ldots \omega_k)(x) \succeq (\omega_1\circ \ldots \omega_k)(y)$. 

Obviously $\succeq \, \subseteq\, \succeq'$, now suppose that $x \succ y$ and suppose by means of contradiction that $y \succeq' x$; it follows that there exist $\omega_1,\ldots,\omega_k\in \mathcal{M}$ for which $(\omega_1\circ\ldots \omega_k)(y) \succeq (\omega_1\circ \ldots\omega_k)(x)$, contradicting the fact that $\succeq$ is $\mathcal{M}$-coherent.  So $\succeq'$ is an extension of $\succeq$.

We claim that $\succeq'$ is transitive.  Suppose that $x \succeq' y \succeq' z$.  Because $x\succeq' y$, there are $\omega_1,\ldots,\omega_k\in \mathcal{M}$ for which $(\omega_1\circ \ldots \circ \omega_k)(x) \succeq (\omega_1\circ \ldots \omega_k)(y)$.  Similarly as $y \succeq' z$, there are $\omega'_1,\ldots,\omega'_l\in \mathcal{M}$ for which $(\omega'_1\circ \ldots \omega'_l)(y)\succeq (\omega'_1\circ \ldots \omega'_l)(z)$.  By coherency of $\succeq$ and by commutativity of $\mathcal{M}$, we may conclude that $(\omega_1\circ \ldots \circ \omega_k\circ \omega'_1 \circ \ldots \omega'_l)(x)\succeq (\omega_1\circ \ldots \circ \omega_k\circ \omega'_1 \circ \ldots \omega'_l)(y)$ and $(\omega_1\circ \ldots \circ \omega_k\circ \omega'_1 \circ \ldots \omega'_l)(y)\succeq (\omega_1\circ \ldots \circ \omega_k\circ \omega'_1 \circ \ldots \omega'_l)(z)$, so that $(\omega_1\circ \ldots \circ \omega_k\circ \omega'_1 \circ \ldots \omega'_l)(x)\succeq (\omega_1\circ \ldots \circ \omega_k\circ \omega'_1 \circ \ldots \omega'_l)(z)$, so that $x \succeq' z$.

We claim that $\succeq’$ is strongly coherent.  First we show that it is coherent.  Suppose that $x \succeq' y$ and let $\omega\in \mathcal{M}$.  Then there are $\omega_1,\ldots,\omega_k\in M$ for which $(\omega_1\circ\ldots \omega_k)(x) \succeq (\omega_1\circ \ldots \omega_k)(y)$.  Since $\mathcal{M}$ is a commutative family and by coherency of $\succeq$, we have $(\omega_1\circ\ldots \omega_k)\big(\omega(x)\big) \succeq (\omega_1\circ \ldots \omega_k)\big(\omega(y)\big)$, so that $\omega(x) \succeq' \omega(y)$.  Suppose that $x \succ' y$ and by means of contradiction that for some $\omega\in \mathcal{M}$, $\omega(y) \succeq' \omega(x)$.  Then by definition there are $\omega'_1 , \ldots \omega_l' \in \mathcal{M}$ for which $(\omega'_1\circ \ldots \omega'_l)\big(\omega(y)\big) \succeq (\omega'_1\circ \ldots \omega'_l)\big(\omega(x)\big)$, which by definition implies that $y  \succeq' x$, a contradiction.  This establishes that $\succeq'$ is coherent.

To see that it is strongly coherent, first suppose that $\omega(x) \succeq' \omega(y)$; then there are $\omega_1,\ldots,\omega_k\in \mathcal{M}$ for which $(\omega_1 \circ \ldots \omega_k \circ \omega)(x) \succeq (\omega_1\circ \ldots \omega_k \circ \omega)(y)$, which implies by definition that $x \succeq' y$.  Suppose in addition $\omega(y) \succeq' \omega(x)$ is false.  By means of contradiction suppose that $y \succeq' x$.  Then again there are $\omega_1,\ldots,\omega_k$ for which $(\omega_1 \circ \ldots \omega_k)(y)\succeq (\omega_1 \circ \ldots \omega_k)(x)$.  By coherence of $\succeq$, $(\omega \circ \omega_1 \ldots \omega_k)(y) \succeq (\omega\circ \omega_1 \ldots \omega_k)(x)$.  By commutativity of $\mathcal{M}$, $(\omega_1 \circ \ldots \omega_k)\big(\omega(y)\big)\succeq (\omega_1 \circ \ldots \omega_k)\big(\omega(x)\big)$, so that $\omega(y) \succeq' \omega(x)$, a contradiction.  \end{proof}

 \begin{lemma}\label{lem:acyclic2}Let $\mathcal{M}$ be a commutative family.  Let $\succeq$ be a $\mathcal{M}$-coherent preorder, and let $w,z\in X$ for which $w\parallel z$ (and in particular $w \neq z$). Then there is an acyclic, $\mathcal{M}$-coherent extension $\succeq’$ of $\succeq$ that renders $w$ and $z$ comparable.\end{lemma}
 
 \begin{proof}
 By appeal to commutativity, any finite string of compositions of functions in $\mathcal{M}$ may be expressed as:
 \[
 f_1^{n_1} \circ f_2^{n_2} \circ \cdots \circ f_L^{n_L},
 \]
 where $\{f_1, \ldots, f_L\} \subseteq \mathcal{M}$, and $n_1,\ldots, n_L \in \mathbb{N}$.\footnote{Note, however, that in general it will be impossible to guarantee a unique representation of this form. For example, suppose $f: X \to X$ is bijective, and $\{f, f^{-1}\}\subseteq \mathcal{M}$.}  Given such an expression, define $\mathbf{f}: \mathcal{M} \to \mathbb{N}_0$ as the unique function such that $f_l \mapsto n_l$ and $g \mapsto 0$ if and only if $g \not \in \{f_1, \ldots, f_L\}$.  
 
 Suppose now, for sake of obtaining a contradiction, that no acyclic, $\mathcal{M}$-coherent extension of $\succeq$ exists that compares $w$ and $z$.  Then every $\mathcal{M}$-coherent extension that renders $w$ and $z$ comparable contains some cycle; in particular, the minimal such extensions obtained  either by adding $w \succ' z$ and $\mathbf{f}(w) \succ' \mathbf{f}(z)$ for all $\mathbf{f}$ associated with some finite composition of elements of $\mathcal{M}$, by adding $z \succ' w$ and all $\mathbf{f}(z) \succ' \mathbf{f}(w)$, or by adding $z \sim' w$ and all $\mathbf{f}(z) \sim' \mathbf{f}(w)$, must contain some cycle.  Consider first $\succeq' \, = \, \succeq \, \cup \; \succeq^*$, where $\succeq^*$ contains all relations of the form $w \succ^* z$ and $\mathbf{f}(w) \succ^* \mathbf{f}(z)$ for all finite compositions of elements of $\mathcal{M}$, $\mathbf{f}$.  It follows there exists a cycle composed of relations of two forms:
\begin{equation}
\begin{aligned}
x \succeq \mathbf{a}^1(w) & \quad & \mathbf{a}^1(w) \succ^* \mathbf{a}^1(z)\\
\mathbf{a}^1(z) \succeq \mathbf{a}^2(w) & \quad & \mathbf{a}^2(w) \succ^* \mathbf{a}^2(z)\\
 \vdots \quad \quad \quad \quad & \quad &  \vdots \quad \quad \quad \quad \\
 \mathbf{a}^{I-1}(z) \succeq \mathbf{a}^I(w) & \quad & \mathbf{a}^I(w) \succ^* \mathbf{a}^I(z)\\
\mathbf{a}^I(z) \succeq x, & \quad & 
\end{aligned}
\end{equation}
for some $x \in X$, where the left column consists of relations in $\succeq$ and the right solely of relations in $\succeq' \setminus \succeq$.  Note that $I \ge 2$, and without loss of generality, each $\mathbf{a}^i$ is distinct.\footnote{If $I=1$, then we have $\mathbf{a}^1(z) \succeq x$ and $x \succeq \mathbf{a}^1(w)$, hence $\mathbf{a}^1(z)\succeq \mathbf{a}^1(w)$. Since $\succeq$ is strongly $\mathcal{M}$-coherent, this implies $w$ and $v$ are $\succeq$-related.}

Analogously, if $\succeq' \; = \; \succeq \; \cup \; \succeq^*$, where $\succeq^*$ contains all relations of the form $z \succ^* w$ and $\mathbf{f}(z) \succ^* \mathbf{f}(w)$ for finite compositions $\mathbf{f}$, then there exists a cycle of the form:
\begin{equation*}
\begin{aligned}
x' \succeq \mathbf{b}^1(z) & \quad & \mathbf{b}^1(z) \succ^* \mathbf{b}^1(w)\\
\mathbf{b}^1(w) \succeq \mathbf{b}^2(z) & \quad & \mathbf{b}^2(z) \succ^* \mathbf{b}^2(w)\\
 \vdots \quad \quad \quad \quad & \quad &  \vdots \quad \quad \quad \quad \\
 \mathbf{b}^{J-1}(w) \succeq \mathbf{b}^J(z) & \quad & \mathbf{b}^J(z) \succ^* \mathbf{b}^J(w)\\
\mathbf{b}^J(w) \succeq x', & \quad & 
\end{aligned}
\end{equation*}
for some $x' \in X$, where again the left column consists of relations in $\succeq$, the right solely of relations in $\succeq' \setminus \succeq$, $J \ge 2$, and each $\mathbf{b}^j$ unique.

Finally, suppose $\succeq' \; = \; \succeq \; \cup \; \succeq^*$, where $\succeq^*$ contains all relations of the form $z \sim^* w$ and $\mathbf{f}(z) \sim^* \mathbf{f}(w)$ for finite compositions $\mathbf{f}$.  By hypothesis, there is a cycle of the form:
 \begin{equation*}
\begin{aligned}
x'' \succeq \mathbf{c}^1(x_1) & \quad & \mathbf{c}^1(x_1) \sim^* \mathbf{c}^1(y_1)\\
\mathbf{c}^1(y_1) \succeq \mathbf{c}^2(x_2) & \quad & \mathbf{c}^2(x_2) \sim^* \mathbf{c}^2(y_2)\\
 \vdots \quad \quad \quad \quad & \quad &  \vdots \quad \quad \quad \quad \\
 \mathbf{c}^{K-1}(y_{K-1}) \succeq \mathbf{c}^J(x_K) & \quad & \mathbf{c}^K(a_K) \sim^* \mathbf{c}^J(y_K)\\
\mathbf{c}^J(y_K) \succeq x'', & \quad & 
\end{aligned}
\end{equation*}
where at least one relation in the left-hand column is strict, $K \ge 2$, each $\mathbf{c}^k$ is unique, and for all $k = 1,\ldots, K$, $\{x_k,y_k\} = \{w,z\}$.   

Define:
\begin{equation*} 
\begin{aligned}
\mathbf{p}^i & = \mathbf{a}^{i+1} - \mathbf{a}^i\\
\mathbf{q}^j & = \mathbf{b}^{j+1} - \mathbf{b}^j\\
\mathbf{r}^k & = \mathbf{c}^{k+1} - \mathbf{c}^k,\\
\end{aligned}
\end{equation*}
where we define indices $I + 1, J+1, K+1 \equiv 1$.  Note that each $\mathbf{p}^i, \mathbf{q}^j,$ and  $\mathbf{r}^k$ is not equal to the zero function $\mathbf{0}$, and:
\begin{equation*}
    \sum_{i =1}^I \mathbf{p}^i  = \sum_{j =1}^J \mathbf{q}^j = \sum_{k =1}^K \mathbf{r}^k  =  \mathbf{0}.
\end{equation*}
Consider the sets:
\begin{equation*}
\begin{aligned}
    \tilde{A}_{wz} & = \big\{\mathbf{r}^k \; \vert \; y_k = w, \; x_{k+1} = z\big\} \\
    \tilde{A}_{zw} & = \big\{\mathbf{r}^k \; \vert \; y_k = z, \; x_{k+1} = w\big\} \\
    \tilde{A}_{ww} & = \big\{\mathbf{r}^k \; \vert \; y_k = w, \; x_{k+1} = w\big\} \\
    \tilde{A}_{zz} & = \big\{\mathbf{r}^k \; \vert \; y_k = z, \; x_{k+1} = z\big\}.
\end{aligned}
\end{equation*}
Clearly these sets cover $\{\mathbf{r}^1, \ldots, \mathbf{r}^K\}$.  Define:
\begin{equation*}
\begin{aligned}
    A_{wz} & = \tilde{A}_{wz}\\
    A_{zw} & = \tilde{A}_{zw} \setminus \tilde{A}_{wz}\\
    A_{ww} & = \tilde{A}_{ww} \setminus \tilde{A}_{zw} \setminus \tilde{A}_{wz}\\
    A_{zz} &= \tilde{A}_{zz} \setminus \tilde{A}_{ww} \setminus \tilde{A}_{zw} \setminus \tilde{A}_{wz},
\end{aligned}
\end{equation*}
if these sets are non-empty, and if empty define them as $\{\mathbf{0}\}$.  By hypothesis, at least some of the sets must contain non-zero elements.
Note that each element of  $\{\mathbf{r}^1, \ldots, \mathbf{r}^K\}$ is contained in exactly one set in the collection $\{A_{wz}, A_{zw}, A_{ww}, A_{zz}\}$.  Let $\{\mathbf{s}^m_{wz}\}_{m=1}^{\vert A_{wz}\vert}$ (resp. $\{\mathbf{s}^m_{zw}\}_{m=1}^{\vert A_{zw}\vert}$, $\{\mathbf{s}^m_{ww}\}_{m=1}^{\vert A_{ww}\vert}$, and $\{\mathbf{s}^m_{zz}\}_{m=1}^{\vert A_{zz}\vert}$) denote enumerations of $A_{wz}$ (resp. $A_{zw}, A_{ww},$ and $A_{zz}$).

We now establish a contradiction, by showing that $\succeq$ contains a cycle, contrary to our hypothesis that it is a preorder.  Let $\bar{\mathbf{h}}$ denote a sufficiently large map $M \to \mathbb{N}_0$ with finite support.\footnote{Sufficiently in the sense only that each vector in the following sequence remain non-negative valued.}  

We will consider two cases in turn.

\textbf{Case 1:  $|A_{wz}|+|A_{zw}|>0$.}

To build our cycle, we first define two chains in $\succeq$ which will prove important in our construction.\footnote{The first chain indexes by $|A_{wz}|$ and the second indexes by $|A_{zw}|$; if either of these are zero, these chains are vacuous.}

\begin{equation*}
    \begin{aligned}
       \bar{\mathbf{h}}(z) & \succeq (\bar{\mathbf{h}} + \mathbf{p}^1)(w) \\
       & \succeq (\bar{\mathbf{h}} + \mathbf{p}^1 + \mathbf{s}^1_{wz})(z) \\
       & \quad \quad \quad \quad \quad \vdots \\
       & \succeq \bigg(\bar{\mathbf{h}} + \vert A_{wz} \vert \sum_{i=1}^I \mathbf{p}^i + I \sum_{m=1}^{\vert A_{wz}\vert} \mathbf{s}^m_{wz}\bigg)(z) \\
       & \quad \quad \quad \quad \quad \vdots \\
       & \succeq \bigg(\bar{\mathbf{h}} + J \, \vert A_{wz} \vert \sum_{i=1}^I \mathbf{p}^i + IJ \sum_{m=1}^{\vert A_{wz}\vert} \mathbf{s}^m_{wz}\bigg)(z).
    \end{aligned}
\end{equation*}

The first part of this chain, up to $\bigg(\bar{\mathbf{h}} + \vert A_{wz} \vert \sum_{i=1}^I \mathbf{p}^i + I \sum_{m=1}^{\vert A_{wz}\vert} \mathbf{s}^m_{wz}\bigg)(z)$ is constructed as follows.  For every $l=1,\ldots,I|A_{wz}|$, every term of the form $(\bar{\mathbf{h}}+\ldots +\mathbf{p}^l)(w)$ is followed by a term of the form $(\bar{\mathbf{h}}+\ldots+\mathbf{p}^l +\mathbf{s}_{wz}^l)(z)$, and for every $l=0,\ldots,I|A_{wz}|-1$, every term of the form $(\bar{\mathbf{h}}+\ldots +\mathbf{s}_{wz}^l)(z)$ is followed by a term of the form $(\bar{\mathbf{h}}+\ldots+\mathbf{p}^l +\mathbf{s}_{wz}^{l+1})(w)$, where an $l$ index on $\mathbf{p}$ is modulo $I$ and on $s_{wz}$ is modulo $|A_{wz}|$.

The second part of this chain, up to $\bigg(\bar{\mathbf{h}} + J \, \vert A_{wz} \vert \sum_{i=1}^I \mathbf{p}^i + IJ \sum_{m=1}^{\vert A_{wz}\vert} \mathbf{s}^m_{wz}\bigg)(z)$, follows by iterating the first $I|A_{wz}|$ steps an additional $|J|-1$ times.

Similarly, there is a chain:
\begin{equation*}
    \begin{aligned}
       \bar{\mathbf{h}}(z) & \succeq (\bar{\mathbf{h}} + \mathbf{s}^1_{zw})(w)\\
       & \succeq (\bar{\mathbf{h}} + \mathbf{s}^1_{zw}+ \mathbf{q}^1)(z) \\
       & \quad \quad \quad \quad \quad \vdots \\
       & \succeq \bigg(\bar{\mathbf{h}} + J \sum_{m=1}^{\vert A_{zw}\vert} \mathbf{s}^m_{wz} + \vert A_{zw} \vert \sum_{j=1}^J \mathbf{q}^j \bigg)(z) \\
       & \quad \quad \quad \quad \quad \vdots \\
       & \succeq \bigg(\bar{\mathbf{h}}  + I J \sum_{m=1}^{\vert A_{zw}\vert} \mathbf{s}^m_{wz} + I \, \vert A_{zw}  \vert \sum_{j=1}^J \mathbf{q}^j\bigg)(z).
    \end{aligned}
\end{equation*}

Appending these chains together then yields a chain:
\begin{equation*}
\bar{\mathbf{h}}(z) \succeq \cdots \succeq \bigg(\bar{\mathbf{h}} + I \, \vert A_{zw}  \vert \sum_{j=1}^J \mathbf{q}^j  + J \, \vert A_{wz} \vert \sum_{i=1}^I \mathbf{p}^i + IJ \sum_{m=1}^{\vert A_{wz}\vert} \mathbf{s}^m_{wz} + I J \sum_{m=1}^{\vert A_{zw}\vert} \mathbf{s}^m_{wz}\bigg)(z).
\end{equation*}

Consider now the following modification to this chain: immediately after the first instance of an $\mathbf{f}(z) \succeq \mathbf{g}(w)$ relation, apply $IJ$ applications of each transformation in $A_{ww}$.  Similarly, after the first $\mathbf{f}(w) \succeq \mathbf{g}(z)$ relation, insert $IJ$ repetitions of each transformation in $A_{zz}$.  The result is a chain:

\begin{equation*}
\bar{\mathbf{h}}(z) \succeq \cdots \succeq \bigg(\bar{\mathbf{h}} + I \, \vert A_{zw}  \vert \sum_{j=1}^J \mathbf{q}^j  + J \, \vert A_{wz} \vert \sum_{i=1}^I \mathbf{p}^i + IJ \sum_{k=1}^K \mathbf{r}^k\bigg)(z).
\end{equation*}

However, since $\sum_i \mathbf{p}^i= \sum_j \mathbf{q}^j = \sum_k \mathbf{r}^k = \mathbf{0}$, this chain is in fact a cycle.  Moreover, since every relation in the left-hand column of ($\ast$) appears in this cycle, it contains at least one strict relation, contradicting the hypothesis that $\succ$ is a preorder. Thus an $\mathcal{M}$-coherent extension of $\succeq$ that compares $w$ and $z$ exists.

\textbf{Case 2:  $|A_{wz}|+|A_{zw}|=0$.}

Follow the idea of Case 1, except here we first construct a single chain of the form:

\begin{equation*}
    \begin{aligned}
       \bar{\mathbf{h}}(z) & \succeq (\bar{\mathbf{h}} + \mathbf{p}^1)(w) \\
       & \succeq (\bar{\mathbf{h}} + \mathbf{p}^1 + \mathbf{q}^1)(z) \\
       & \quad \quad \quad \quad \quad \vdots \\
       & \succeq \bigg(\bar{\mathbf{h}} + J  \sum_{i=1}^I \mathbf{p}^i + I \sum_{j=1}^{J} \mathbf{q}^j\bigg)(z) \\
    \end{aligned}
\end{equation*}

Consider now the following modification to this chain: immediately after the first instance of an $\mathbf{f}(z) \succeq \mathbf{g}(w)$ relation, apply an application of each transformation in $A_{ww}$.  Similarly, after the first $\mathbf{f}(w) \succeq \mathbf{g}(z)$ relation, insert an application of each transformation in $A_{zz}$.  The result is a chain:

\begin{equation*}
\bar{\mathbf{h}}(z) \succeq \cdots \succeq \bigg(\bar{\mathbf{h}} + I \sum_{j=1}^J \mathbf{q}^j  + J  \sum_{i=1}^I \mathbf{p}^i + \sum_{k=1}^K \mathbf{r}^k\bigg)(z).
\end{equation*}

Conclude similarly to Case 1.

\end{proof}

The remainder of the proof uses a standard transfinite induction argument.  Given is our $\mathcal{M}$-coherent preorder.  Say that $\succeq_1\rhd \succeq_2$ whenever $\succeq_1$ extends $\succeq_2$.  By Lemma~\ref{lem:strongcoherent2}, assume without loss that $\succeq^*$ is a strongly $\mathcal{M}$-coherent preorder.  First, for any chain of strongly $M$-coherent preorder extensions, $\{\succeq_{\lambda}\}_{\Lambda}$, it follows that $\overline{\succeq}=\bigcup_{\Lambda}\succeq_{\lambda}$ is also a strongly $\mathcal{M}$-coherent preorder extension.  This is standard:  If $x \succeq^* y$, then likewise $x\overline{\succeq}y$ and if in addition, $x\succ^* y$, then for no $\lambda$ is it the case that $y\succeq_{\lambda} x$.  So $x\overline{\succ}y$.  Transitivity is standard.  

Now we show that $\overline{\succeq}$ is coherent.  Suppose that $(x,y)\in\overline{\succeq}$.  Then there is $\lambda\in\Lambda$ for which $x \succeq_{\lambda}y$, since $\succeq_{\lambda}$ is $\mathcal{M}$-coherent, we conclude $\omega(x)\succeq_{\lambda}\omega(y)$ and consequently $\omega(x)\overline{\succeq}\omega(y)$.  Likewise, if $x \overline{\succ} y$, then for all $\lambda$, $(y,x)\notin \succeq_{\lambda}$, and consequently, $(\omega(y),\omega(x))\notin\succeq_{\lambda}$ by strong $\mathcal{M}$-coherency.   So $\omega(x) \overline{\succ} \omega(y)$.  

Finally, we want to show that $\overline{\succeq}$ is also strongly coherent; so let $x,y\in X$ and $\omega\in \mathcal{M}$ for which $\omega(x) \overline{\succeq}\omega(y)$.  Then there is $\lambda$ for which $\omega(x) \succeq_{\lambda}\omega(y)$ so that $x \succeq_{\lambda}y$ by strong coherency of $\succeq_{\lambda}$.  Further if $\omega(x)\overline{\succ}\omega(y)$ then for no $\lambda$ is it the case that $\omega(y) \succeq_{\lambda} \omega(x)$.  By coherency of each $\succeq_{\lambda}$, for no $\lambda$ is it the case that $y\succeq_{\lambda}x$, so that $x \overline{\succ} y$.

By Zorn's Lemma, there is a maximal strongly coherent preorder extension $\succeq’$.  We claim that it is complete. Suppose by means of contradiction that there are $z,w$ which are unranked.  By Lemma~\ref{lem:acyclic2}, there is an acyclic, $\mathcal{M}$-coherent extension of $\succeq’$ which strictly ranks $z$ and $w$.  By Lemma~\ref{lem:transitive}, the transitive closure of this extension satisfies $\mathcal{M}$-coherency, and is a preorder extension by acyclicity.  Finally by Lemma~\ref{lem:strongcoherent2} there is a strongly $\mathcal{M}$-coherent preorder extension of this relation which then ranks $z$ and $w$ and therefore extends $\succeq’$, contradicting maximality of $\succeq’$.

\end{appendix}

\bibliographystyle{ecta}
\bibliography{notes}
\end{document}